\RequirePackage{fix-cm}
\documentclass[smallextended]{svjour3}       
\smartqed  
\usepackage{graphicx}

\usepackage{amsmath,amscd,amssymb,latexsym}

\newcommand{\gf}{{\mathrm{GF}}}

\begin{document}

\title{On $-1$-differential uniformity of ternary APN power functions \thanks{H. Yan is with School of Mathematics, Southwest Jiaotong University, Chengdu, 610031, China;
and is also with Guangxi Key Laboratory of Cryptography and Information Security, Guilin, 541000, China
(email: hdyan@swjtu.edu.cn).}
}


\author{Haode Yan
}

\institute{
}

\date{Received: date / Accepted: date}

\maketitle

\begin{abstract}
Very recently, a new concept called multiplicative differential and the corresponding $c$-differential uniformity were introduced by Ellingsen \textit{et al}. A function $F(x)$ over finite field $\gf(p^n)$ to itself is called $c$-differential uniformity $\delta$, or equivalent, $F(x)$ is differentially $(c,\delta)$ uniform, when the maximum number of solutions $x\in\gf(p^n)$ of $F(x+a)-F(cx)=b$, $a,b,c\in\gf(p^n)$, $c\neq1$ if $a=0$, is equal to $\delta$. The objective of this paper is to study the $-1$-differential uniformity of ternary APN power functions $F(x)=x^d$ over $\gf(3^n)$. We obtain ternary power functions with low $-1$-differential uniformity, and some of them are almost perfect $-1$-nonlinear.
\end{abstract}

\section{Introduction}

Differential cryptanalysis (\cite{BS,BS93}) is one of the most fundamental cryptanalytic approaches targeting symmetric-key primitives. Such a cryptanalysis approach has attracted a lot of attention because it was proposed to be the first statistical attack for breaking the iterated block ciphers \cite{BS}. The security of cryptographic functions regarding differential attacks was widely studied in the past 30 years. This type of security is quantified by the so-called \emph{differential uniformity} of the substitution box (S-box) used in the cipher \cite{NK}. In \cite{BCJW}, a new type of differential was proposed. The authors utilized modular multiplication as a primitive operation, which extends the type of differential cryptanalysis. It is necessary to start the theoretical analysis of an (output) multiplicative differential. Motivated by practical differential cryptanalysis, Ellingsen \textit{et al.} recently coined a new concept called \emph{multiplicative differential } and the corresponding $c$-differential uniformity (\cite{EFRST}).

\begin{definition}Let $\gf(p^n)$ denote the finite field with $p^n$ elements, where $p$ is a prime number and $n$ is a positive integer. For a function $F$ from $\gf(p^n)$ to itself, $a,c \in \gf(p^n)$, the (multiplicative) $c$ derivative of $F$ with respect to $a$ is define as
\[_cD_aF(x)=F(x+a)-cF(x), ~\mathrm{for}~\mathrm{all}~x.\]
For $b\in\gf(p^n)$, let $_c\Delta_F(a,b)=\#\{x\in\gf(p^n): F(x+a)-cF(x)=b\}$. We call $_c\Delta_F=\mathrm{max}\{_c\Delta_F(a,b):a,b\in\gf(p^n), \mathrm{and}~ a\neq 0 ~\mathrm{if}~ c=1\}$ the $c$-differential uniformity of F. If $_c\Delta_F=\delta$, then we say $F$ is differentially $(c,\delta)$-uniform.
\end{definition}
If the $c$-differential uniformity of $F$ equals $1$, then $F$ is called a perfect $c$-nonlinear (P$c$N) function. P$c$N functions over odd characteristic finite fields are also called $c$-planar functions. If the $c$-differential uniformity of $F$ is $2$, then $F$ is called an almost perfect $c$-nonlinear (AP$c$N) function. It is easy to see that, for $c=1$ and $a\neq0$, the $c$-differential uniformity becomes the usual differential uniformity, and the P$c$N and AP$c$N functions become perfect nonlinear (PN) function and almost perfect nonlinear function (APN) respectively. These functions are of great significance in both theory and practical applications. For even characteristic finite fields, APN functions have the lowest differential uniformity. Known APN functions over even characteristic finite fields were presented in \cite{BD,D1,D2,D3,G,JW,K,N}. For the known results on PN and APN functions over odd characteristic finite fields, the readers are referred to \cite{CM,DMMPW,DO,DY,HRS,HS,L,ZW10,ZW11}.

Because of the strong resistance to differential attacks and the low implementation cost in a hardware environment, power function $F(x)=x^d$ (i.e., monomials) with low differential uniformity can serve as a good candidate for the design of S-boxes. Moreover, power functions with low differential uniformity may also introduce some unsuitable weaknesses within a cipher \cite{BCC,JaKn,CoPi,CaVi}. For instance, a differentially $4$-uniform power function, which is extended affine EA-equivalent to the inverse function $x \mapsto x^{2^n-2}$ over $\gf(2^n)$ with even $n$, is employed in the AES (advanced encryption standard). A nature question one would ask is whether the power functions have good $c$-differential properties. In \cite{EFRST}, the authors studied the $c$-differential uniformity of the well-known inverse function $F(x)=x^{p^n-2}$ over $\gf(p^n)$ for both even and odd prime $p$. It was shown  that $F(x)$ is P$c$N when $c=0$, $F(x)$ is AP$c$N with some conditions of $c$ and $F(x)$ is differentially $(c,3)$-uniform otherwise. This result illustrates that P$c$N functions can exist for $p=2$. For P$c$N functions $x^{\frac{3^k+1}{2}}$ over $\gf(3^n)$ and $c=-1$, a sufficient and necessary condition was presented in \cite{EFRST}. In \cite{BT}, it was shown that for odd $p$, $n$ and $c=-1$, $x^{\frac{p^2+1}{2}}$ over $\gf(p^n)$ and $x^{p^2-p+1}$ over $\gf(p^3)$ are P$c$N functions. In \cite{YMZ}, it was proved that the Gold function over even characteristic finite field is differentially $(c,3)$-uniform for $c\neq1$. Some P$c$N and AP$c$N functions were also obtained. Moreover, for $c$-differential uniformity of power function $F(x)=x^d$ over $\gf(p^n)$ with $c\neq1$, the following lemma was introduced.

\begin{lemma}[\cite{YMZ}]\label{power}Let $F(x)=x^d$ be a power function over $\gf(p^n)$. Then
$$
_c\Delta_F=\mathrm{max}\big\{~ \{{_c\Delta_F}(1,b): b\in\gf(p^n) \} \cup \{\gcd(d,p^n-1)\}~\big\}.
$$
\end{lemma}

\begin{table}[t]\label{table-1}
\caption{Power functions $F(x)=x^d$ over $\gf(p^n)$ with  low $c$-differential uniformity}
\centering
\begin{tabular}{|c||c|c|c|c|}
\hline
$p$&$d$ & condition & $_c\Delta_F$ & References \\
[0.5ex]
\hline
any& $2$& $c\neq1$ & 2 &\cite{EFRST}\\
\hline
any &$p^n-2$ &$c=0$ & $1$ &\cite{EFRST}\\
\hline
2 &$2^n-2$ &$c\neq0$, $\mathrm{Tr_n}(c)=\mathrm{Tr_n}(c^{-1})=1$ & $2$ &\cite{EFRST}\\
\hline
2 &$2^n-2$ &$c\neq0$, $\mathrm{Tr_n}(c)=0$ or $\mathrm{Tr_n}(c^{-1})=0$ & $3$ &\cite{EFRST}\\
\hline
odd &$p^n-2$ &$c=4$, $c=4^{-1}$ or $\chi(c^2-4c)=\chi(1-4c)=-1$  & $2$ &\cite{EFRST}\\
\hline
odd &$p^n-2$ &$c\neq0,4,4^{-1}$, $\chi(c^2-4c)$=1 or $\chi(1-4c)=1$ & $3$ &\cite{EFRST}\\
\hline
3& $({3^k+1})/{2}$& $c=-1$, $n/\gcd(k,n)=1$ & 1 &\cite{EFRST}\\
\hline
odd& $({p^2+1})/{2}$& $c=-1$, $n$ odd & 1 &\cite{BT}\\
\hline
odd& $p^2-p+1$& $c=-1$, $n=3$ & 1 &\cite{BT}\\
\hline
2& $2^k+1$& $c\neq1$, $\gcd(k,n)=1$& 3 &\cite{YMZ}\\
\hline
odd & $p^k+1$& $1\neq c\in\gf(p)$, $\gcd(k,n)=1$& 2 &\cite{YMZ}\\
\hline
odd& $(p^k+1)/2$& $c=-1$, $k/\gcd(k,n)$ is even& $1$ &\cite{YMZ}\\
\hline
3& $(3^k+1)/2$& $c=-1$, $k$ odd, $\gcd(k,n)=1$& $2$ &\cite{YMZ}\\
\hline
any& $(2p^n-1)/3$& $c\neq1$, $p^n\equiv 2 (\mathrm{mod}~3)$& $\leq 3$ &\cite{YMZ}\\
\hline
odd& $(p^n+1)/2$& $c\neq\pm1$ & $\leq 4$ &\cite{YMZ}\\
\hline
odd& $(p^n+1)/2$& $c\neq\pm1$, $\chi(\frac{1-c}{1+c})=1$, $p^n\equiv 1 (\mathrm{mod}~4)$ & $\leq 2$ &\cite{YMZ}\\
\hline
$>3$& $(p^n+3)/2$&  $c=-1$, $p^n\equiv 3 (\mathrm{mod}~4)$& $\leq 3$ &\cite{YMZ}\\
\hline
$>3$& $(p^n+3)/2$& $c=-1$, $p^n\equiv 1 (\mathrm{mod}~4)$& $\leq 4$ &\cite{YMZ}\\
\hline
odd& $(p^n-3)/2$& $c=-1$ & $\leq 4$ &\cite{YMZ}\\
\hline
\end{tabular}
\begin{itemize}
    \item $\mathrm{Tr_n}(\cdot)$ denotes the absolute trace mapping from $\gf(2^n)$ to $\gf(2)$.
    \item $\chi(\cdot)$ denotes the quadratic multiplicative character on $\gf(p^n)^*$.
\end{itemize}

\end{table}

\begin{table}[t]\label{table-2}
\caption{Results in this paper}
\centering
\begin{tabular}{|c||c|c|c|c|}
\hline
$p$&$d$ & condition & $_c\Delta_F$  \\
[0.5ex]

\hline
3& $(3^{\frac{n+1}{2}}-1)/2$& $c=-1$, $n\equiv 1 (\mathrm{mod}~4)$ & $\leq 2$ \\
\hline

3& $(3^{\frac{n+1}{2}}-1)/2+(3^n-1)/2$& $c=-1$, $n\equiv 3 (\mathrm{mod}~4)$ & $\leq 2$\\
\hline

3& $(3^{n+1}-1)/8$& $c=-1$, $n\equiv 1 (\mathrm{mod}~4)$ & $\leq 2$\\
\hline

3& $(3^{n+1}-1)/8+(3^n-1)/2$& $c=-1$, $n\equiv 3 (\mathrm{mod}~4)$ & $\leq 2$\\
\hline
3& $(3^{\frac{n+1}{2}}-1)/2$& $c=-1$, $n\equiv 3 (\mathrm{mod}~4)$ & $\leq 4$\\
\hline
3& $(3^{\frac{n+1}{2}}-1)/2+(3^n-1)/2$& $c=-1$, $n\equiv 1 (\mathrm{mod}~4)$ & $\leq 4$ \\
\hline
3& $(3^{n+1}-1)/8$& $c=-1$, $n\equiv 3 (\mathrm{mod}~4)$ & $\leq 4$\\
\hline
3& $(3^{n+1}-1)/8+(3^n-1)/2$& $c=-1$, $n\equiv 1 (\mathrm{mod}~4)$ & $\leq 4$\\
\hline
3& ${(3^{\frac{n+1}{4}}-1)(3^{\frac{n+1}{2}}+1)}$& $c=-1$, $n\equiv 3 (\mathrm{mod}~4)$ & $\leq 4$\\
\hline
3& $(3^n+1)/4+(3^n-1)/2$& $c=-1$, $n$ odd & $\leq 4$\\
\hline
\end{tabular}

\end{table}

As summarized in Table \ref{table-1}, $c=-1$ is a very special case and sometimes the $-1$-differential uniformity is lower than $c$-differential uniformity for other $c\in\gf(p^n)$. The perfect $-1$-nonlinear function was also called quasi-planar function \cite{BT}. In this paper, we study the $-1$-differential uniformity of $F(x)$ when $F(x)$ is a ternary APN power function. Lemma \ref{power} indicates that to determine the $-1$-differential uniformity of power functions, the following $-1$-differential equation needs to be studied.
\[\Delta(x)=(x+1)^d+x^d=b.\]
Let $\delta(b)=\#\{x\in\gf(3^n)~|~\Delta(x)=b\}$. The maximum value of $\{\delta(b)~|~b\in\gf(3^n)\}$ plays an important role in studying the $-1$-differential uniformity of $F(x)$. In the rest of this paper, we consider several classes of ternary APN power functions. It turns out that they are with low $-1$-differential uniformity, and some of them are almost perfect $-1$-nonlinear. The results in this paper are shown in Table \ref{table-2}I.

\section{$-1$-differential uniformity of $x^{\frac{3^{\frac{n+1}{2}}-1}{2}}$ over $\gf(3^n)$}
In this section, let $F(x)=x^d$ be a power function over $\gf(3^n)$, where $n \equiv 1 (\mathrm{mod}~4)$ and $d=\frac{1}{2}(3^{\frac{n+1}{2}}-1)$. It was proved in \cite{DMMPW} that $F(x)$ is an APN function. We consider the $-1$-differential uniformity of $F(x)$ as follows.
\begin{theorem}\label{m}Let $F(x)=x^d$ be a power function over $\gf(3^n)$, where $n \equiv 1 (\mathrm{mod}~4)$ and $d=\frac{1}{2}(3^{\frac{n+1}{2}}-1)$. We have $_{-1}\Delta_F \leq 2$.
\end{theorem}

\begin{proof} Let $m=\frac{n+1}{2}$. Note that $2(3^m+1)d-3(3^n-1)=2$ and $d$ is odd when $n \equiv 1 (\mathrm{mod}~4)$, then $\gcd(d,3^n-1)=1$, i.e., $F(x)$ is a permutation on $\gf(3^n)$. For $b\in\gf(3^n)$, we consider the $-1$-differential equation
\begin{equation}\label{meqn1}
\Delta(x)=(x+1)^d+x^d=b.
\end{equation}
Let $u_{x+1}=(x+1)^d$ and $u_x=x^d$. For $x\neq0$, note that
\begin{equation}\label{meqn2}
u^{3^m+1}_x=x^{\frac{3^{n+1}-1}{2}}=\chi(x)x.
\end{equation}
Herein and hereafter, let $\chi$ denote the quadratic multiplicative character on $\gf(3^n)^*$. Let $x\in\gf(3^n)\setminus\{0,-1\}$ be a solution of (\ref{meqn1}) for fixed $b\in\gf(3^n)$, then $u_x,u_{x+1}\neq0$. Taking the $(3^m+1)$th power on both sides of $u_{x+1}=-u_x+b$, we have
\begin{equation}\label{meqn3}
bu^{3^m}_x+b^{3^m}u_x=-\chi(x+1)(x+1)+\chi(x)x+b^{3^m+1}.
\end{equation}
For equation (\ref{meqn3}), we distinguish the following four cases.

Case I. $\chi(x+1)=\chi(x)=1$.

In this case, we have $bu^{3^m}_x+b^{3^m}u_x=b^{3^m+1}-1$ from (\ref{meqn3}). Since the mapping $u_x\mapsto bu^{3^m}_x+b^{3^m}u_x$ is bijective on $\gf(3^n)$, we can find a unique $u_x$. Because $F(x)$ is a permutation, a unique $x$ can be found from the $u_x$. This case has at most one solution.

Case II. $\chi(x+1)=\chi(x)=-1$.

This case has at most one solution. The discussion is similar to that of Case I and we omit it.

Case III. $\chi(x+1)=1, \chi(x)=-1$. From (\ref{meqn3}), we have $bu^{3^m}_x+b^{3^m}u_x=x+b^{3^m+1}-1$ in this case, and then we have
\begin{equation}\label{meqn4}
(u_x+b)^{3^m+1}=-b^{3^m+1}-1
\end{equation}
by (\ref{meqn2}). If there are two distinct solutions in this case, namely $x_3$ and $x'_3$, then $u_{x_3}$ and $u_{x'_3}$ satisfy (\ref{meqn4}) with $\chi(x_3+1)=\chi(x'_3+1)=1$ and $\chi(x_3)=\chi(x'_3)=-1$. Consequently $(u_{x_3}+b)^{3^m+1}=(u_{x'_3}+b)^{3^m+1}$ can be obtained from (\ref{meqn4}). Then we have $u_{x_3}+b=-(u_{x'_3}+b)$ since $\gcd(3^m+1,3^n-1)=2$ and $x_3\neq x'_3$, which leads to $u_{x_3}=b-u_{x'_3}=u_{x'_3+1}$. However, the above conclusion contradicts to $\chi(u_{x_3})=\chi(x_3)=-1$ and $\chi(u_{x'_3+1})=\chi(x'_3+1)=1$. We conclude that Case III has at most one solution.

Case IV. $\chi(x+1)=-1, \chi(x)=1$. In this case we have $bu^{3^m}_x+b^{3^m}u_x=-x+b^{3^m+1}+1$ from (\ref{meqn3}), and then
\begin{equation}\label{meqn5}
(u_x+b)^{3^m+1}=-b^{3^m+1}+1
\end{equation}
by (\ref{meqn2}). Similar to Case III, we can obtain that this case has at most one solution.

Next we will prove that for fixed $b$, (\ref{meqn1}) cannot have solution in Case I and Case II simultaneously. Otherwise, suppose that $x_1$ and $x_2$ are solutions of (\ref{meqn1}) in Case I and Case II with $\chi(x_1+1)=\chi(x_1)=1$ and $\chi(x_2+1)=\chi(x_2)=-1$ respectively. Then we have $bu^{3^m}_{x_1}+b^{3^m}u_{x_1}=b^{3^m+1}-1$ and $bu^{3^m}_{x_2}+b^{3^m}u_{x_2}=b^{3^m+1}+1$, where $u_{x_1}$ and $u_{x_2}$ we defined before. Now we have $b(u_{x_1}+u_{x_2})^{3^m}+b^{3^m}(u_{x_1}+u_{x_2})=-b^{3^m+1}$ and the consequent $u_{x_1}+u_{x_2}=b$. From (\ref{meqn1}), we can obtain $u_{x_2}=u_{x_1+1}$, which contradicts to $\chi(u_{x_2})=\chi(x_2)=-1$ and $\chi(u_{x_1+1})=\chi(x_1+1)=1$. Therefore, we conclude that (\ref{meqn1}) has at most one solution in Cases I and II for fixed $b\in\gf(3^n)$.

Then we prove that for fixed $b$, (\ref{meqn1}) cannot have solution in Case III and Case IV simultaneously. Otherwise, suppose that $x_3$ and $x_4$ are solutions of (\ref{meqn1}) in Case III and Case IV with $\chi(x_3+1)=1$, $\chi(x_3)=-1$ and $\chi(x_4+1)=-1, \chi(x_4)=1$ respectively. Then $x_3$ and $x_4$ satisfy (\ref{meqn4}) and (\ref{meqn5}) respectively. By the sum of (\ref{meqn4}) and (\ref{meqn5}), we have
\begin{equation}\label{meqn6}
(u_{x_3}+b)^{3^m+1}+(u_{x_4}+b)^{3^m+1}=b^{3^m+1}.
\end{equation}
Taking the $3^m$th power on both sides of (\ref{meqn6}), we have
\begin{equation}\label{meqn7}
(u_{x_3}+b)^{3^m+3}+(u_{x_4}+b)^{3^m+3}=b^{3^m+3}
\end{equation}
since $3^m(3^m+1)=3^{n+1}+3^m=3^m+3+3(3^n-1)$. From (\ref{meqn6}) and (\ref{meqn7}), we have $(u_{x_3}+b)^{3^m+3}+(u_{x_4}+b)^{3^m+3}=b^2(u_{x_3}+b)^{3^m+1}+b^2(u_{x_4}+b)^{3^m+1}$, that is
\begin{equation}\label{meqn8}
-(u_{x_3}+b)^{3^m+1}u_{x_3}(b-u_{x_3})=(u_{x_4}+b)^{3^m+1}u_{x_4}(b-u_{x_4}).
\end{equation}
Note that $b-u_{x_3}=u_{x_3+1}$ and $b-u_{x_4}=u_{{x_4}+1}$, the left-hand side of (\ref{meqn8}) is a square element and the right-hand side of (\ref{meqn8}) is a nonsquare element. Then $u_{x_3}+b=u_{x_4}+b=0$ can be obtained, i.e. $u_{x_3}=u_{x_4}$, which contradicts to $\chi(u_{x_3})=\chi(x_3)=-1$ and $\chi(u_{x_4})=\chi(x_4)=1$. We conclude that (\ref{meqn1}) has at most one solution in Cases III and IV for fixed $b\in\gf(3^n)$.

From the above discussions, (\ref{meqn1}) has at most two solutions in $\gf(3^n)\setminus\{0,-1\}$. One can be easily calculate that $\Delta(0)=1$ and $\Delta(-1)=-1$. For $b=1$ and $b=-1$, it can be verified that $\Delta(x)=1$ and $\Delta(x)=-1$ has no solution in $\gf(3^n)\setminus\{0,-1\}$, i,e, $\delta(1)=\delta(-1)=1$. Then we obtain $\delta(b)\leq 2$ for any $b$, which leads to $_{-1}\Delta_F\leq2$ by Lemma \ref{power} and $\gcd(d,3^n-1)=1$ .

\end{proof}

For $n \equiv 3 (\mathrm{mod}~4)$, we can also get power functions with low $-1$-differential uniformity.
\begin{theorem}Let $F(x)=x^d$ be a power function over $\gf(3^n)$, where $n \equiv 3 (\mathrm{mod}~4)$ and $d=\frac{1}{2}(3^{\frac{n+1}{2}}-1)$. We have $_{-1}\Delta_F \leq 4$.
\end{theorem}

\begin{proof}The proof is similar to that of Theorem \ref{m}. We give a sketch here. In this case, $\gcd(d,3^n-1)=2$. With the notation we used before, equations (\ref{meqn1}), (\ref{meqn2}) and (\ref{meqn3}) also hold. Let $x\in\gf(3^n)\setminus\{0,-1\}$ be a solution of (\ref{meqn3}) for fixed $b\in\gf(3^n)$, four cases are considered as follows.

Case I. $\chi(x+1)=\chi(x)=1$.

In this case, We have $bu^{3^m}_x+b^{3^m}u_x=b^{3^m+1}-1$ from (\ref{meqn3}). Since the mapping $u_x\mapsto bu^{3^m}_x+b^{3^m}u_x$ is bijective on $\gf(3^n)$, we can find a unique $u_x$. Then a unique $x$ can be found for $\chi(x)=1$. This case has at most one solution.

Case II. $\chi(x+1)=\chi(x)=-1$.

Similar to Case I, this case has at most one solution.

Case III. $\chi(x+1)=1, \chi(x)=-1$. We have $bu^{3^m}_x+b^{3^m}u_x=x+b^{3^m+1}-1$ from (\ref{meqn3}), and then we have $(u_x+b)^{3^m+1}=-b^{3^m+1}-1$ by (\ref{meqn2}).
We can obtain two $u_x$'s since $\gcd(d,3^n-1)=2$ and the consequent two $x$'s for given $\chi(x)$. This case has at most two solutions.

Case IV. $\chi(x+1)=-1, \chi(x)=1$.

Similar to Case III, this case has at most two solutions.

One can similarly prove that for fixed $b$, (\ref{meqn1}) cannot have solution in Case III and Case IV simultaneously. By discussions as above, we know that (\ref{meqn1}) has at most four solutions in $\gf(3^n)\setminus\{0,-1\}$. We have $\Delta(0)=\Delta(-1)=1$. For $b=1$, one can easily verify that $\Delta(x)=1$ has no solution in $\gf(3^n)\setminus\{0,-1\}$, i.e., $\delta(1)=2$. Then we obtain $\delta(b)\leq 4$ for any $b$, this leads to $_{-1}\Delta_F\leq4$ by Lemma \ref{power} and $\gcd(d,3^n-1)=2$.

\end{proof}

For $d'=d+\frac{3^n-1}{2}$, we have the following corollary.
\begin{corollary}Let $F'(x)=x^{d'}$ be a power function over $\gf(3^n)$, where $n$ is an odd integer and $d'=\frac{3^{\frac{n+1}{2}}-1}{2}+\frac{3^n-1}{2}$. We have $_{-1}\Delta_{F'}\leq 2$ when $n\equiv 3 (\mathrm{mod} ~4)$ and $_{-1}\Delta_{F'}\leq 4$ when $n\equiv 1 (\mathrm{mod} ~4)$.
\end{corollary}
\begin{proof}It can be calculated that $\gcd(d',3^n-1)\leq 2$. First we consider $n\equiv 3 (\mathrm{mod} ~4)$, i.e., $3n\equiv 1 (\mathrm{mod} ~4)$. By Theorem \ref{m},
\begin{equation}\label{meqn9}
(x+1)^{\frac{3^{\frac{3n+1}{2}}-1}{2}}+x^{\frac{3^{\frac{3n+1}{2}}-1}{2}}=b
\end{equation}
has at most two solutions in $\gf(3^{3n})$ for any $b\in\gf(3^{3n})$. Since $(3^n-1)|\frac{3^{\frac{3n+1}{2}}-1}{2}-d'$, equation (\ref{meqn9}) becomes $(x+1)^{d'}+x^{d'}=b$ any $x,b\in\gf(3^n)$. Therefore, this equation has at most two solutions in $\gf(3^n)$, i.e., $_{-1}\Delta_{F'}\leq2$. The other case can be proved similarly and we omit the details.
\end{proof}
\section{$-1$-differential uniformity of $x^{\frac{3^{n+1}-1}{8}}$ over $\gf(3^n)$}

In this section, let $F(x)=x^d$ be a power function over $\gf(3^n)$, where $n \equiv 1 (\mathrm{mod}~4)$ and $d=\frac{3^{n+1}-1}{8}$. It was proved in \cite{DMMPW} that $F(x)$ is an APN function. We consider the $-1$-differential uniformity of $F(x)$ as follows.

\begin{theorem}\label{8}Let $F(x)=x^d$ be a power function over $\gf(3^n)$, where $n \equiv 1 (\mathrm{mod}~4)$ and $d=\frac{3^{n+1}-1}{8}$. We have $_{-1}\Delta_F \leq 2$.
\end{theorem}

\begin{proof}Note that $\gcd(d,3^n-1)=1$, $F(x)$ is a permutation on $\gf(3^n)$. For $b\in\gf(3^n)$, we consider the $c$-differential equation
\begin{equation}\label{8eqn1}
\Delta(x)=(x+1)^d+x^d=b.
\end{equation}
Let $u_{x+1}=(x+1)^d$ and $u_x=x^d$. For $x\neq0$, note that
\begin{equation}\label{8eqn2}
u^4_x=x^{4d}=\chi(x)x.
\end{equation}
Let $x\in\gf(3^n)\setminus\{0,-1\}$ be a solution of (\ref{8eqn1}) for fixed $b\in\gf(3^n)$, then $u_{x},u_{x+1}\neq0$. Taking the $4$th power on both sides of $u_{x+1}=-u_x+b$, we have
\begin{equation}\label{8eqn3}
bu^3_x+b^3u_x=-\chi(x+1)(x+1)+\chi(x)x+b^4.
\end{equation}
For (\ref{8eqn3}), we distinguish the following four cases.

Case I. $\chi(x+1)=\chi(x)=1$.

In this case, we have $bu^3_x+b^3u_x=b^4-1$ from (\ref{8eqn3}). Since the mapping $u_x \mapsto bu^3_x+b^3u_x$ is bijective on $\gf(3^n)$, we can find a unique $u_x$. Because $F(x)$ is a permutation, a unique $x$ can be found from the $u_x$. This case has at most one solution.

Case II. $\chi(x+1)=\chi(x)=-1$.

This case has at most one solution. The discussion is similar to that of Case I and we omit it.

Case III. $\chi(x+1)=1, \chi(x)=-1$. From (\ref{8eqn3}), we have $bu^3_x+b^3u_x=x+b^4-1$ in this case, and then we have
\begin{equation}\label{8eqn4}
(u_x+b)^4=-b^4-1
\end{equation}
by (\ref{8eqn2}). If there are two distinct solutions in this case, namely $x_3$ and $x'_3$, then $u_{x_3}$ and $u_{x'_3}$ satisfy (\ref{8eqn4}) with $\chi(x_3+1)=\chi(x'_3+1)=1$ and $\chi(x_3)=\chi(x'_3)=-1$. Consequently  $(u_{x_3}+b)^4=(u_{x'_3}+b)^4$ can be obtained from (\ref{8eqn4}). Then we have $u_{x_3}+b=-(u_{x'_3}+b)$ since $x_3\neq x'_3$, which leads to $u_{x_3}=b-u_{x'_3}=u_{x'_3+1}$. However, the above conclusion contradicts to $\chi(u_{x_3})=\chi(x_3)=-1$ and $\chi(u_{x'_3+1})=\chi(x'_3+1)=1$. We conclude that Case III has at most one solution.

Case IV. $\chi(x+1)=-1, \chi(x)=1$. In this case we have $bu^3_x+b^3u_x=-x+b^4+1$ from (\ref{8eqn3}), and then
\begin{equation}\label{8eqn5}
(u_x+b)^4=-b^4+1
\end{equation}
by (\ref{8eqn2}). Similar to Case III, we can obtain that this case has at most one solution.

Next we will prove for fixed $b$, (\ref{8eqn1}) cannot have solutions in Case I and Case II simultaneously. Otherwise, suppose that $x_1$ and $x_2$ are solutions of (\ref{8eqn1}) in Case I and Case II with $\chi(x_1+1)=\chi(x_1)=1$ and $\chi(x_2+1)=\chi(x_2)=-1$ respectively. Then we have $bu^3_{x_1}+b^3u_{x_1}=b^4-1$ and $bu^3_{x_2}+b^3u_{x_2}=b^4+1$, where $u_{x_1}$ and $u_{x_2}$ we defined before. Now we have $b(u_{x_1}+u_{x_2})^3+b^3(u_{x_1}+u_{x_2})=-b^4$ and the consequent $u_{x_1}+u_{x_2}=b$. From (\ref{8eqn1}), we can obtain $u_{x_2}=u_{x_1+1}$, which contradicts $\chi(u_{x_2})=\chi(x_2)=-1$ and $\chi(u_{x_1+1})=\chi(x_1+1)=1$. Therefore, we conclude that (\ref{8eqn1}) has at most one solution in Cases I and II for fixed $b\in\gf(3^n)$.

Then we prove for fixed $b$, (\ref{8eqn1}) cannot have solutions in Case III and Case IV simultaneously. Otherwise, suppose that $x_3$ and $x_4$ are solutions of (\ref{8eqn1}) in Case III and Case IV with $\chi(x_3+1)=1, \chi(x_3)=-1$ and $\chi(x_4+1)=-1, \chi(x_4)=1$ respectively. Then $x_3$ and $x_4$ satisfy (\ref{8eqn4}) and (\ref{8eqn5}) respectively. By the sum of (\ref{8eqn4}) and (\ref{8eqn5}), we have $(u_{x_3}+b)^4+(u_{x_4}+b)^4=b^4$, that is,
\begin{equation}\label{8eqn6}
(u^2_{x_3}-bu_{x_3}+u^2_{x_4}-bu_{x_4}+b^2)^2=-u_{x_3}(b-u_{x_3})u_{x_4}(b-u_{x_4}).
\end{equation}
Note that $b-u_{x_3}=u_{x_3+1}$ and $b-u_{x_4}=u_{{x_4}+1}$, the right-hand side of (\ref{8eqn6}) is a nonzero nonsquare element, which is a contradiction. We conclude that (\ref{8eqn1}) has at most one solution in Cases III and IV for fixed $b\in\gf(3^n)$. From the above discussions, (\ref{8eqn1}) has at most two solutions in $\gf(3^n)\setminus\{0,-1\}$.

One can easily calculate that $\Delta(0)=1$ and $\Delta(-1)=-1$. For $b=1$ and $b=-1$, it can be verified that $\Delta(x)=1$ and $\Delta(x)=-1$ has no solution in $\gf(3^n)\setminus\{0,-1\}$, i,e, $\delta(1)=\delta(-1)=1$. Then we obtain $\delta(b)\leq 2$ for any $b$, this leads to $_{-1}\Delta_F\leq2$ by Lemma \ref{power} and $\gcd(d,3^n-1)=1$ .

\end{proof}

For $n \equiv 3 (\mathrm{mod}~4)$ and $d'=d+\frac{3^n-1}{2}$, we list the following theorems without proof.
\begin{theorem}Let $F(x)=x^d$ be a power function over $\gf(3^n)$, where $n \equiv 3 (\mathrm{mod}~4)$ and $d=\frac{3^{n+1}-1}{8}$. We have $_{-1}\Delta_F \leq 4$.
\end{theorem}

\begin{corollary}Let $F'(x)=x^{d'}$ be a power function over $\gf(3^n)$, where $n$ is an odd integer and $d'=\frac{3^{n+1}-1}{8}+\frac{3^n-1}{2}$. We have $_{-1}\Delta_{F'}\leq 2$ when $n\equiv 3 (\mathrm{mod} ~4)$ and $_{-1}\Delta_{F'}\leq 4$ when $n\equiv 1 (\mathrm{mod} ~4)$.
\end{corollary}

\section{$-1$-differential uniformity of $x^{\frac{3^n+1}{4}+\frac{3^n-1}{2}}$ over $\gf(3^n)$}

It was proved in \cite{HRS} that the power function $x^d$ is an APN function over $\gf(3^n)$, where $n$ is an odd integer and $d=\frac{3^n+1}{4}+\frac{3^n-1}{2}$. The $-1$-differential uniformity is considered as follows.
\begin{theorem}\label{ternaryfromapn}Let $F(x)=x^d$ be a power function over $\gf(3^n)$, where $n$ is odd and $d=\frac{3^n+1}{4}+\frac{3^n-1}{2}$. Then $_{-1}\Delta_F\leq 4$.
\end{theorem}

\begin{proof}One can easily obtain that $d$ is even and $\gcd(d,3^n-1)=2$. Note that $\chi(-1)=-1$ since $n$ is odd. We consider the $-1$-differential equation
\begin{equation}\label{eqnternaryfromapn}
\Delta(x)=(x+1)^d+x^d=b.
\end{equation}
When $b=0$, (\ref{eqnternaryfromapn}) has no solution. For fixed $b\in\gf(3^n)^*$, let $x\in\gf(3^n)\setminus\{0,-1\}$ is a solution of (\ref{eqnternaryfromapn}), we distinguish the following four cases.

Case I.  $\chi(x+1)=\chi(x)=1$. Let $x+1=\alpha^2$ and $x=\beta^2$ for $\alpha,\beta\in\gf(3^n)^*$, then $\alpha^2-\beta^2=1$. We can obtain $\chi(\alpha)\alpha+\chi(\beta)\beta=b$ from (\ref{eqnternaryfromapn}). We have
\[\beta^2+1=\alpha^2=(\chi(\alpha)\alpha)^2=(b-\chi(\beta)\beta)^2=b^2+b\chi(\beta)\beta+\beta^2.\]
One can obtain $\chi(\beta)\beta=b^{-1}-b$ and $x=\beta^2=(\chi(\beta)\beta)^2=(b^{-1}-b)^2$. This case has at most one solution.

Case II. $\chi(x+1)=\chi(x)=-1$. Let $x+1=-\alpha^2$ and $x=-\beta^2$ for $\alpha,\beta\in\gf(3^n)^*$, then $\alpha^2-\beta^2=-1$. Similar to Case I, we can obtain $x=-(b+b^{-1})^2$. This case has at most one solution.

Case III. $\chi(x+1)=1, \chi(x)=-1$. Let $x+1=\alpha^2$ and $x=-\beta^2$ for $\alpha,\beta\in\gf(3^n)^*$, then $\alpha^2+\beta^2=1$. We can obtain $\chi(\alpha)\alpha+\chi(\beta)\beta=b$ from (\ref{eqnternaryfromapn}). Let $\gamma=\chi(\beta)\beta$, which is a square element in $\gf(3^n)$. Then $\gamma^2=\beta^2$ and $\gamma$ satisfies $(b-\gamma)^2+\gamma^2=1$, i.e.
\begin{equation}\label{gamma1}
\gamma^2-b\gamma+1-b^2=0,
\end{equation}
which is a quadratic equation on $\gamma$. Equation (\ref{gamma1}) has most two solutions, then we can obtain at most two $x$'s since $x=-\gamma^2$. This case has at most two solutions.

Case IV. $\chi(x+1)=-1, \chi(x)=1$. Let $x+1=-\alpha^2$ and $x=\beta^2$ for $\alpha,\beta\in\gf(3^n)^*$, then $\alpha^2+\beta^2=-1$. We can obtain $\chi(\alpha)\alpha+\chi(\beta)\beta=b$ from (\ref{eqnternaryfromapn}). Let $\gamma=\chi(\beta)\beta$, which is a square element in $\gf(3^n)$. Then $\gamma^2=\beta^2$ and $\gamma$ satisfies $(b-\gamma)^2+\gamma^2=-1$, i.e.
\begin{equation}\label{gamma2}
\gamma^2-b\gamma-1-b^2=0,
\end{equation}
which is a quadratic equation on $\gamma$. Equation (\ref{gamma2}) has most two solutions, then we can obtain at most two $x$'s since $x=\gamma^2$. This case has at most two solutions.

Note that $x$ is a solution of (\ref{eqnternaryfromapn}) if and only if $-x-1$ is a solution of (\ref{eqnternaryfromapn}). This implies that when $x\neq 1$ (the corresponds $b\neq-1$), if Case III (the same for Case IV) has solutions, it must has two solutions.  Next we prove that for fixed $b\in\gf(3^n)\setminus\{0,-1\}$, (\ref{eqnternaryfromapn}) cannot have solution in Case III and Case IV simultaneously. Suppose on the contrary that $x_1$, $x_2$ are distinct solutions of (\ref{eqnternaryfromapn}) for some given $b$ in Case III, and $x_3$, $x_4$ are distinct solutions of (\ref{eqnternaryfromapn}) for the same $b$ in Case IV. By the discussions above, each $x_i, 1\leq i\leq4$ corresponds to square element $\gamma_i$. Moreover, $\gamma_1$, $\gamma_2$ are the two solutions of (\ref{gamma1}), and $\gamma_3$, $\gamma_4$ are the two solutions of (\ref{gamma2}). They satisfy $\gamma_1+\gamma_2=\gamma_3+\gamma_4=b$, $\gamma_1\gamma_2=1-b^2$ and $\gamma_3\gamma_4=-1-b^2$. We can obtain
$\gamma^2_1+\gamma^2_2+\gamma^2_3+\gamma^2_4=(\gamma_1+\gamma_2)^2+\gamma_1\gamma_2+(\gamma_3+\gamma_4)^2+\gamma_3\gamma_4=0$.
Since $\gamma_4\neq0$, let $\delta_i=\gamma_i/\gamma_4$, $1\leq i \leq 3$, then $\delta_1,\delta_2$ and $\delta_3$ are square elements, and they satisfy $\delta_1+\delta_2-\delta_3-1=0$ and $\delta^2_1+\delta^2_2+\delta^2_3+1=0$. Replace by $\delta_3=\delta_1+\delta_2-1$, we have the following quadratic equation on $\delta_1$.
\[\delta^2_1-(\delta_2-1)\delta_1+(\delta^2_2-\delta_2+1)=0.\]
The discriminate of the above quadratic equation is $\Delta=(\delta_2-1)^2-(\delta^2_2-\delta_2+1)=-\delta_2$, which is a nonzero nonsquare element in $\gf(3^n)$. It contradicts to $\delta_1\in\gf(3^n)$. Then we proved that for $b\in\gf(3^n)^*$, (\ref{eqnternaryfromapn}) has at most $4$ solutions in $\gf(3^n)\setminus\{0,-1\}$.

One can easily check that $\Delta(0)=\Delta(-1)=1$. For $b=1$, it can be verified that $\Delta(x)=1$ has no solution in the four cases, i.e., $\Delta(x)=1$ has no solution in $\gf(3^n)\setminus\{0,-1\}$, $\delta(1)=2$. For $b=-1$, it can be verified that $x=1$ is the only solution of (\ref{eqnternaryfromapn}), i.e. $\delta(-1)=1$. This with the discussions above leads to $_{-1}\Delta_F\leq 4$.
\end{proof}

\section{$-1$-differential uniformity of $x^{(3^{\frac{n+1}{4}}-1)(3^{\frac{n+1}{2}}+1)}$ over $\gf(3^n)$}
In \cite{ZW11}, the authors studied the power function $F(x)=x^d$ over $\gf(3^n)$, where $n \equiv 3 (\mathrm{mod}~4)$ and $d=(3^{\frac{n+1}{4}}-1)(3^{\frac{n+1}{2}}+1)$. It was shown that $x^d$ is an APN function.
In what follows, we discuss the $-1$-differential uniformity of $F(x)$.
\begin{theorem}Let $F(x)=x^d$ be a power function over $\gf(3^n)$, where $n \equiv 3 (\mathrm{mod}~4)$, $d=(3^m-1)(3^{2m}+1)$ and $m=\frac{n+1}{4}$. Then $_{-1}\Delta_F\leq 4$.
\end{theorem}

\begin{proof}Note that $d$ is an even number and $\gcd(d,3^n-1)=2$. For $b\in\gf(3^n)$, we consider equation
\begin{equation}\label{zha+main}\Delta(x)=(x+1)^d+x^d=b.
\end{equation}
It is easy to see that $\Delta(0)=\Delta(-1)=1$, and (\ref{zha+main}) has no solution when $b=0$. Let $x\in\gf(3^n)\setminus\{0,-1\}$ be a solution of (\ref{zha+main}) for some given $b\in\gf(3^n)^*$. Denote by $u_{x+1}=(x+1)^d$ and $u_x=x^d$. Since $\frac{3^m+1}{2}\cdot d=\frac{3^{n+1}-1}{2}\equiv 1+\frac{3^n-1}{2} (\mathrm{mod}~3^n-1)$, we have ${u_x}^{\frac{3^m+1}{2}}=\chi(x)x$ and ${u_{x+1}}^{\frac{3^m+1}{2}}=\chi(x+1)(x+1)$. One can easily see that if $u_x$ and $\chi(x)$ are given, $x$ can be determined uniquely.

Let $\xi\in\gf(3^{2n})\setminus\{0,\pm1\}$ such that $\frac{u_x}{b}=\xi+\frac{1}{\xi}-1=\frac{(\xi+1)^2}{\xi}$, then we have $\frac{u_{x+1}}{b}=-\xi-\frac{1}{\xi}-1=-\frac{(\xi-1)^2}{\xi}$ by (\ref{zha+main}). Moreover, we can obtain $\chi(x)x={u_x}^{\frac{3^m+1}{2}}=(\frac{b(\xi+1)^2}{\xi})^{\frac{3^m+1}{2}}$ and $\frac{b(\xi+1)^2}{\xi}=x^d=(\frac{b(\xi+1)^2}{\xi})^{\frac{3^m+1}{2}\cdot d}$.
Similarly, $-\frac{b(\xi-1)^2}{\xi}=(x+1)^d=(-\frac{b(\xi-1)^2}{\xi})^{\frac{3^m+1}{2}\cdot d}$. Then $\xi$ satisfies $-(\frac{\xi+1}{\xi-1})^2=(\frac{\xi+1}{\xi-1})^{(3^m+1)d}$, i.e., $(\frac{\xi+1}{\xi-1})^{3(3^n-1)}=-1$. This with $\xi\in\gf(3^{2n})$ leads to $\xi^{3^n+1}=1$. In the following, we discuss equation (\ref{zha+main}) in two cases.

Case 1. $\chi(x+1)=\chi(x)$.

In this case, ${u_{x+1}}^{\frac{3^m+1}{2}}-{u_x}^{\frac{3^m+1}{2}}=\chi(x+1)(x+1)-\chi(x)x=\chi(x)$. That is,
\[(-\frac{b(\xi-1)^2}{\xi})^{\frac{3^m+1}{2}}-(\frac{b(\xi+1)^2}{\xi})^{\frac{3^m+1}{2}}=\chi(x).\]
We deduce the following equation
\begin{equation}\label{zha+case1}(-1)^{\frac{3^m+1}{2}}(\xi-1)^{3^m+1}-(\xi+1)^{3^m+1}=\chi(x)b^{-\frac{3^m+1}{2}}\xi^{\frac{3^m+1}{2}}.
\end{equation}
Two subcases are considered as follows.

Subcase 1.1. $\frac{3^m+1}{2}$ is even, i.e., $m$ is odd. Then (\ref{zha+case1}) becomes $\xi^{3^m}+\xi=\chi(x)b^{-\frac{3^m+1}{2}}\xi^{\frac{3^m+1}{2}}$. Let $t=\xi^{\frac{3^m-1}{2}}$, then $t_{1,2}=-\chi(x)b^{-\frac{3^m+1}{2}}\pm\sqrt{b^{-(3^m+1)}-1}$.
Since $m$ is odd, then $\gcd(m,2n)=1$ and $\gcd(\frac{3^m-1}{2},3^{2n}-1)=1$. We can obtain a unique $\xi_1$ from $\xi^{\frac{3^m-1}{2}}=t_1$ since $\gcd(\frac{3^m-1}{2},3^{2n}-1)=1$. For $t_2=t^{-1}_1$, we can also obtain a unique $\xi_2$ such that $\xi_2^{\frac{3^m-1}{2}}=t_2$. Note that $\xi_2=\xi^{-1}_1$ and they give the same $u_x$.


Subcase 1.2. $\frac{3^m+1}{2}$ is odd, i.e., $m$ is even. Then (\ref{zha+case1}) becomes $\xi^{3^m+1}+1=\chi(x)b^{-\frac{3^m+1}{2}}\xi^{\frac{3^m+1}{2}}$. Let $t=\xi^{\frac{3^m+1}{2}}$, then $t_{1,2}=-\chi(x)b^{-\frac{3^m+1}{2}}\pm\sqrt{b^{-(3^m+1)}-1}$. Since $m$ is even, then $\gcd(m,2n)=2$ and $\gcd(\frac{3^m+1}{2},3^{2n}-1)=1$. We can obtain a unique $\xi_1$ from $\xi^{\frac{3^m+1}{2}}=t_1$ since $\gcd(\frac{3^m+1}{2},3^{2n}-1)=1$. For $t_2=t^{-1}_1$, we can also obtain a unique $\xi_2$ such that $\xi_2^{\frac{3^m+1}{2}}=t_2$. Note that $\xi_2=\xi^{-1}_1$ and they give the same $u_x$.

By discussions in the above two subcases, we conclude that one can obtain a unique $u_x$ from given $b$ and $\chi(x)$, and then we find at most one solution of (\ref{zha+case1}) for each $\chi(x)$. This case has at most $2$ solutions.

Case 2. $\chi(x+1)=-\chi(x)$.

In this case, ${u_{x+1}}^{\frac{3^m+1}{2}}+{u_x}^{\frac{3^m+1}{2}}=\chi(x+1)(x+1)+\chi(x)x=-\chi(x)$. That is,
\[(-\frac{b(\xi-1)^2}{\xi})^{\frac{3^m+1}{2}}+(\frac{b(\xi+1)^2}{\xi})^{\frac{3^m+1}{2}}=-\chi(x).\]
We deduce the following equation
\begin{equation}\label{zha+case2}(-1)^{\frac{3^m+1}{2}}(\xi-1)^{3^m+1}+(\xi+1)^{3^m+1}=-\chi(x)b^{-\frac{3^m+1}{2}}\xi^{\frac{3^m+1}{2}}.
\end{equation}
We have following two subcases.

Subcase 2.1. $\frac{3^m+1}{2}$ is even. Then (\ref{zha+case2}) becomes
\[\xi^{3^m+1}+1=\chi(x)b^{-\frac{3^m+1}{2}}\xi^{\frac{3^m+1}{2}}.\]
Let $t=\xi^{\frac{3^m+1}{2}}$, if $\chi(x)=1$, then $t_{1,2}=-b^{-\frac{3^m+1}{2}}\pm\sqrt{b^{-(3^m+1)}-1}$. Note that $t_2=t^{-1}_1$ and they give the same $u_x$'s, we only consider $t_1$. Since $\frac{3^m+1}{2}$ is even, $m$ is odd, then $\gcd(\frac{3^m+1}{2},3^{n}+1)=2$. We can obtain two solutions, namely $\pm\xi_1$, from $\xi^{\frac{3^m+1}{2}}=t_1$ since $\gcd(\frac{3^m+1}{2},3^{n}+1)=2$. If $\chi(x)=-1$, then $t_{3,4}=b^{-\frac{3^m+1}{2}}\pm\sqrt{b^{-(3^m+1)}-1}$. We only consider $t_4$, which satisfies $t_4=-t_1$. Similarly, we obtain another two $\xi$'s, namely $\delta\xi_1$, $-\delta\xi_1$, where $\delta\in\gf(3^{2n})$ with $\delta^2=-1$. In this subcase, we get four distinct $\xi$'s and each of them corresponds a possible solution of (\ref{zha+main}).

Subcase 2.2. $\frac{3^m+1}{2}$ is odd. Then (\ref{zha+case2}) becomes
\[\xi^{3^m}+\xi=\chi(x)b^{-\frac{3^m+1}{2}}\xi^{\frac{3^m+1}{2}}.\]
Let $t=\xi^{\frac{3^m-1}{2}}$, if $\chi(x)=1$, then $t_{1,2}=-b^{-\frac{3^m+1}{2}}\pm\sqrt{b^{-(3^m+1)}-1}$. We only consider the equation $\xi^{\frac{3^m-1}{2}}=t_1$ since the solutions of another equation correspond the same $u_x$'s. Since $\frac{3^m+1}{2}$ is odd, $m$ is even, and $\gcd(\frac{3^m-1}{2},3^{n}+1)=4$. We obtain four solutions from $\xi^{\frac{3^m-1}{2}}=t_1$, namely $\xi_2,\delta\xi_2,-\xi_2,-\delta\xi_2$, where $\delta\in\gf(3^{2n})$ with $\delta^2=-1$. If $\chi(x)=-1$, then $t_{3,4}=b^{-\frac{3^m+1}{2}}\pm\sqrt{b^{-(3^m+1)}-1}$. We only consider $t_4$, which satisfies $t_4=-t_1$. If $\xi'_2$ is a solution of $\xi^{\frac{3^m-1}{2}}=t_4=-t_1$, then $(\frac{\xi'_2}{\xi_2})^{\frac{3^m-1}{2}}=-1$. We obtain that $(\frac{\xi'_2}{\xi_2})^4=1$ from $\gcd(3^m-1,3^{n}+1)=4$, i.e., $\xi'_2=\delta^i\xi_2$, $0\leq i \leq 3$. That means $\xi^{\frac{3^m-1}{2}}=t_4$ cannot contribute new $\xi$'s. We also obtain four distinct $\xi$'s in this subcase.

Recall that $u_{x}=b(\xi+\frac{1}{\xi}-1)$, in the following we prove that $\xi$ and $\delta\xi$ cannot contribute solutions of (\ref{zha+main}) simultaneously, where $\delta$ we defined before. More precisely, let
$u_{x_1}=b(\xi+\frac{1}{\xi}-1)$ and $u_{x_2}=b(\delta\xi+\frac{1}{\delta\xi}-1)$, then we have
\[(u_{x_1}+b)^2+(u_{x_2}+b)^2=b^2((\xi+\frac{1}{\xi})^2+(\delta\xi+\frac{1}{\delta\xi})^2)=b^2.\]
The above identity can be rewritten as
\[(u_{x_1}+u_{x_2}+b)^2=-u_{x_1}u_{x_2},\]
which is a contradiction. That means each of subcases 2.1 and 2.2 has at most two solutions.

By discussions as above, we conclude that $_{-1}\Delta_F\leq 4$. The proof is finished.
\end{proof}

\section{concluding remarks}
In this paper, we studied the $-1$-differential uniformity of ternary APN power functions. We obtain many classes of power functions with low $-1$-differential uniformity, and some of them are almost perfect $-1$-nonlinear. It is mentioned that in this paper we give the upper bound of the $-1$-differential uniformity of some power functions, it is better to study whether the equality holds. In this paper, we only studied $c=-1$, it is also good to study the $c$-differential properties for $\pm1\neq c\in \gf(3^n)$. Our future work is to find more power functions with low $c$-differential uniformity. This topic is widely open. Power functions with low usual differential uniformity are useful in sequences, coding theory, and combinatorial designs. It is worth finding the applications of power functions with low $c$-differential uniformity in such areas.
\section{Acknowledgments}
H. Yan's research was supported by the National Natural Science Foundation of
China Grant (No.11801468) and the Guangxi Key Laboratory of Cryptography and Information Security Grant (No. GCIS201814).

\end{document}